\documentclass[a4paper,10pt]{article}
\usepackage[utf8]{inputenc}
\usepackage{amsthm,amssymb,amsmath}
\newtheorem{theorem}{Theorem}
\newtheorem{lemma}[theorem]{Lemma}

\newtheorem{question}[theorem]{Question}

\author{Matthieu Rosenfeld\\
\small LIRMM, CNRS, Université de Montpellier\\[-0.8ex]
\small\tt matthieu.rosenfeld@gmail.com}
\title{It is undecidable whether the growth rate of a given bilinear system is 1}
\begin{document}

\maketitle

\begin{abstract}
We show that there exists no algorithm that decides for any bilinear system $(B,v)$ if the growth rate of $(B,v)$ is $1$. This answers a question of Bui who showed that if the coefficients are positive the growth rate is computable (i.e., there is an algorithm that outputs the sequence of digits of the growth rate of $(B,v)$). Our proof is based on a reduction of the computation of the joint spectral radius of a set of matrices to the computation of the growth rate of a bilinear system. We also use our reduction to deduce that there exists no algorithm that approximates the growth rate of a bilinear system with relative accuracy $\varepsilon$ in time polynomial in the size of the system and of $\varepsilon$.
Our two results hold even if all the coefficients are nonnegative rationals.
\end{abstract}

\section{Introduction}
Given a bilinear map $B:\mathbb{R}^n\times\mathbb{R}^n\rightarrow \mathbb{R}^n$ and an initial vector $v\in \mathbb{R}^n$, the growth rate of \emph{the bilinear system} $(B,v)$ quantifies how large a vector can be if it is obtained by applying $t-1$ instances of $B$ to $t$ instances of $v$. This quantity was recently used by Rote to provide a bound on the number of minimal dominating sets in any given tree \cite{introJSR}. I then generalized his approach for other sets \cite{RosMSO}. The idea behind these two results is to provide a dynamic algorithm that counts the number of desired sets, to deduce from this algorithm a bilinear system whose growth rate is exactly the desired bound, and then to compute this growth rate. It motivated a few questions about the growth rate of bilinear systems: is it always well-defined? finite? computable? efficiently approximable? algebraic if the coefficients are algebraic/rationals?

For any integer $t\ge1$, let $A_t(B,v)$ be the set of vectors obtained by applying $t-1$ instances of $B$ to $t$ instances of $v$, that is $A_1(B,v)=\{v\}$ and
$$A_t(B,v)=\bigcup_{1\le m\le t-1}\left\{B(x,y):x\in A_m(B,v), y\in A_{t-m}(B,v)\right\}\,.$$
For all $t$, we let $\lambda_t(B,v)=\max\limits_{x\in A_t(B,v)} ||x||$ and \emph{the growth rate of $(B,v)$} is given by$$\lambda(B,v)= \limsup_{t\rightarrow\infty}\lambda_t(B,v)^{1/t} =\limsup_{t\rightarrow\infty}\max\limits_{x\in A_t(B,v)} ||x||^{1/t}$$
where $||.||$ is any norm over $\mathbb{R}^n$. The definition of $\lambda(B,v)$ does not depend on the choice of the norm since all norms over $\mathbb{R}^n$ are equivalent. For the sake of notation, we allow ourselves to write $A_t$, $\lambda_t$ and $\lambda$ without specifying $B$ and $v$ when they are both clear from the context.

We say that a bilinear system $(B,v)$ is \emph{positive} if all the coefficients of $B$ are nonnegative and all the coefficients of $v$ are positive. It is \emph{nonnegative} if all the coefficients of $B$ and $v$ are nonnegative.
In \cite{buigbm1}, Bui showed, amongst other things, that if we use a $\lim$ instead of a $\limsup$ in the definition of the growth rate this quantity remains well-defined for positive bilinear systems (and he provides an example with some non-positive entries where the quantity is not well-defined with the $\lim$).\footnote{In fact, in \cite{buigbm1,buigbm2} the definition of the growth of a bilinear system is given with the $\lim$ instead of a $\limsup$. We stick here to the definition with the $\limsup$ which is more general. In particular, for the applications in \cite{RosMSO} we can restrict ourselves to nonnegative bilinear operators, but we cannot assume that the growth rate is always defined if we use the $\lim$.} It implies that the growth rate of any bilinear system is finite (i.e., the $\limsup$ is never $+\infty$).
In \cite{buigbm2}, he showed that for any positive bilinear system $(B,v)$ there exists positive reals $r$, $C$ and $C'$ such that for all $t$,
\[ Ct^{-r}\lambda^t\le\lambda_t \le C't^{r}\lambda^t\,.\]
The constant $r$, $C$ and $C'$ can be computed from $(B,v)$ and for any integer $t$, $\lambda_t$ can be computed. It provides a theoretical way to approximate $\lambda$ for positive bilinear systems since for all $t$
\begin{equation}\label{buybound}
\left(\frac{t^{-r}\lambda_t}{C'}\right)^{1/t}\le\lambda \le \left(\frac{t^{r}\lambda_t}{C}\right)^{1/t}\,.
\end{equation}
It implies that the growth rate of positive bilinear systems is computable (in the sense that there is an algorithm that from the input $(B,v)$ can output the sequence of digits of $\lambda(B,v)$). It gives a partial answer to one of the main questions raised in \cite{introJSR} and \cite{RosMSO}.
In practice, it is not clear that the bounds on  $r$, $C$, and $C'$ are good enough to be able to obtain any meaningful approximation in a reasonable time (notice that the size of $A_n$ usually grows exponentially fast). Moreover, in the context of \cite{RosMSO} the bilinear system are nonnegative, but not necessarily positive. In particular, it is not particularly hard to come up with examples from this context where equation \eqref{buybound} does not hold.

On the other hand, the technique used in \cite{introJSR,RosMSO} relies on finding a polytope which is a fixed point of $B$ and that contains $\frac{v}{\alpha}$ to deduce that the growth rate is at most $\alpha$. We used it to show that the growth rate of a nonnegative bilinear system is upper semi-computable (i.e., there exists an algorithm that provides a sequence of upper bounds converging toward $\lambda$). 

Bui asked the following question:
\begin{question}\label{maisnquest}
  Given any real number $\alpha$ and any bilinear system $(B,v)$, is the problem of checking if $\alpha=\lambda(B,v)$ decidable?
\end{question}
In this article, we show that it is not the case even if the bilinear system is nonnegative, but the question remains open for positive bilinear systems. We summarize the aforementioned computability results in table \ref{CompGrowthTable}.
\begin{table}[h]
  \begin{tabular}{|c||c|c|c|}\hline
    Bilinear system is &Positive &  Nonnegative  & any\\\hline\hline
    $\lambda(B,V)$  is upper semi-computable & Yes \cite{RosMSO}&Yes \cite{RosMSO}&\\\hline
    $\lambda(B,V)$  is computable & Yes \cite{buigbm2}&&\\\hline
    $\lambda(B,V)=1$ is decidable &&No [this article]& No [this article]\\\hline
  \end{tabular}
  \label{CompGrowthTable}
  \caption{The computability of the growth of bilinear systems.}
\end{table}

 We use the fact that the computation of the joint spectral radius of a set of matrices can be reduced to the computation of the growth rate of a bilinear system. The computability of the joint spectral radius has already been well studied and we can use previous results to deduce a negative answer to Bui's question. We also obtain that there is no polynomial-time algorithm that provides an approximation of the growth rate with bounded relative error (under the assumption that $P\not=NP$).

In section\ref{JSRsec}, we recall some facts about the joint spectral radius. In Section \ref{thereduction}, we show Theorem \ref{thmJSRvsGR}, that is, the joint spectral radius of a set of matrices can be easily expressed as a function of the growth rate of a few bilinear systems and these bilinear systems are easy to compute.
From this, we deduce in Section \ref{secmainres} our two main results: Theorem \ref{thmundecidableGR} and Theorem \ref{thmNPcomplete}. We then conclude the article by mentioning a few open questions and a remark.

For any vector $x\in\mathbb{R}^n$ and any integers $1\le i<j\le n$, we let $x_i$ be the $i$th coordinate of $x$ and $x_{[i\ldots j]}$ be the vector of dimension $j-i+1$ such that  $$x_{[i\ldots j]} = \begin{pmatrix}
        x_i\\x_{i+1}\\\vdots\\  x_j\end{pmatrix}\,.$$

\section{Joint spectral radius}\label{JSRsec}
The notion of joint spectral radius of a set of matrices was introduced by Rota and Strang in the 1960s \cite{introJSR}. This notion generalizes the notion of spectral radius of a single matrix to a set of matrices. We direct the interested reader to a survey on this subject \cite{JSRJunger}.

Let $\Sigma$ be a finite\footnote{The definition of the joint spectral radius is in fact the same for infinite sets of matrices, but for our purpose we only need it on finite sets.} set of matrices over the same vector space, then for any $t\ge 1$, we let $\Sigma^t$ be the set of matrices obtained by taking the product of $t$ matrices from $\Sigma$, that is
$$\Sigma^t=\{M_1M_2\ldots M_t:M_i\in\Sigma\}\,.$$
Then the \emph{joint spectral radius} of $\Sigma$ is defined as the limit
$$\rho(\Sigma)= \lim_{t\rightarrow\infty} \max_{M\in \Sigma^t} ||M||^{1/t}$$
where $||.||$ is any submultiplicative norm (that is, for every matrices $A$ and $B$, $||AB||\le||A||\cdot||B||$). Once again, the choice of the norm does not matter for the value of $\rho(\Sigma)$ as long as the norm is submultiplicative (e.g., every norm induced by a vector norm is submultiplicative).

We now recall two properties of the joint spectral radius. The first result will be useful to show that there is no algorithm deciding equality of the growth rate of a bilinear operator with any positive constant.

\begin{theorem}[{{\cite{undecidableJSR} and \cite[Theorem 2.6]{JSRJunger}}}]\label{thmundecidableJSR}
The problem of determining, given a set of matrices $\Sigma$ , if $\rho(\Sigma)\le1$ (or if $\rho(\Sigma)=1$) is Turing-undecidable.
These results remain true even if $\Sigma$ contains only nonnegative rational entries.
\end{theorem}
The statement from \cite{undecidableJSR} does not include the problem $\rho(\Sigma)=1$. However, the joint spectral radius can be approximated with arbitrary precision, it is not hard to deduce that, the decidability of $\rho(\Sigma)=1$ would imply the decidability of $\rho(\Sigma)\le1$ (if you know that $\rho(\Sigma)\not=1$ you only need a good enough approximation to deduce which strict inequality holds).

The second result will be useful to show that there is no polynomial algorithm that provides an approximation of arbitrary relative accuracy.

\begin{theorem}[{{\cite{NPcompleteJSR} and \cite[Theorem 2.4]{JSRJunger}}}] \label{thmnonpolyJSR}
Unless $P=NP$, there is no algorithm that, given a set of matrices $\Sigma$ and a
relative accuracy $\varepsilon$, returns an estimate $\tilde\rho$ of $\rho(\Sigma)$ such that $|\tilde\rho-\rho|\le \varepsilon\rho$ in a number of steps that is polynomial in the size of $\Sigma$ and $\varepsilon$. It is true even if the matrices in $\Sigma$ have binary entries.
\end{theorem}

\section{The reduction}\label{thereduction}
Let $n$ be a positive integer and $\Sigma\subseteq \mathbb{R}^{n\times n}$ be a non-empty finite set of matrices.
There exists $m\ge2$ and $M_0,\ldots,M_{2^m-1}\in\mathbb{R}^{n\times n}$ such that $\Sigma=\{M_0,\ldots,M_{2^m-1}\}$ (we can always repeat the last matrix so that the total number of matrices is a power of two).

Let $S_i$ be the sequence such that $S_0=3$ and for all $i$, $S_i=S_{i-1}+2^i$ (i.e., $S_i=2^{i+1}+1$). Let $k=S_{m}-1$.
We let $\pi:\mathbb{R}^{k+n}\rightarrow\mathbb{R}^n$
be the projection that maps any vectors $x\in\mathbb{R}^{k+n}$ to the vector made of the last $n$ coordinates of $x$, that is, for all $i\in\{1,\ldots, n\}$, $\pi(x)_i=x_{k+i}$.

We are now ready to give the definition of the bilinear map that we will use in our reduction.
Let $B:\mathbb{R}^{k+n}\times\mathbb{R}^{k+n}\rightarrow\mathbb{R}^{k+n}$ be the bilinear map such that, for all $x,y\in \mathbb{R}^{k+n}$,
\begin{itemize}
  \item $B(x,y)_1=0$,
  \item $B(x,y)_2=x_1\cdot y_1$,
  \item $B(x,y)_3 = x_1\cdot y_2$,
  \item $B(x,y)_4 = x_2\cdot y_1$,
  \item for all $i\in\{1,\ldots, m-1\}$ and $j\in\{0,\ldots, 2^{i}-1\}$, 
  \begin{align*}
    B(x,y)_{2j+S_i}&=x_{j+S_{i-1}}\cdot y_{3}\,,\\   
    B(x,y)_{2j+1+S_i}&=x_{j+S_{i-1}}\cdot y_{4}\,,                                                            
  \end{align*}
  \item and the remaining coordinates are given by
  $$\pi(B(x,y))=\sum_{i=0}^{2^m-1}y_{i+S_{m-1}}\cdot M_{i} \pi(x)\,.$$
\end{itemize}

Let $e_1,\ldots, e_{k+n}$ (resp. $e'_1,\ldots, e'_{n}$) be the canonical basis of $\mathbb{R}^{k+n}$ (resp. $\mathbb{R}^{n}$), where $e_i$ (resp. $e'_i$) is the vector with $1$ in the ith coordinate and $0$'s elsewhere.

Our goal is to show, that for any $i\in\{1,\ldots,n\}$, we can ``simulate'' repeated applications of matrices from $\Sigma$ to $e'_i$ by repeated applications of $B$ to $e_1+e_{k+i}$. It is done in Lemmas \ref{controleBitsLemma} and \ref{eqbilivsli} and it allows us to conclude in Lemma \ref{rhoaslambda} that we can reduce the computation of the joint spectral radius of any set of matrices to the computation of the growth rate of some bilinear systems. We first try to provide some informal intuition in the next paragraph.

We constructed $B$ in such a way that with the initial vector $v=e_1+e_{k+i}$, the first $k$ coordinates play the role of a ``logical controller'' and the last $n$ coordinate are where the ``real computation'' happens. In fact, every vector  (except for the initial vector $v$) has either exactly $1$ of the $k$ first coordinates at $1$ and all the other coordinates at $0$, or all the first $k$ coordinates at $0$. For any $i\le k$, there is a way to construct $e_i$ (this can be seen by considering the first 5 items of the definition of $B$) and for all $i\in \{0,\ldots, 2^m-1\}$ and $x\in \mathbb{R}$,
$$\pi(B(x,e_{S_{m-1}+i}))=M_i\pi(x)$$
(this is a direct consequence of the last item of the definition of $B$).
Hence, we can construct any $e_i$ and use it to multiply the $n$ last coordinates of another vector by the desired matrix. Moreover, trying to do anything else with $B$ and $v$ only results in the vector $0$. Finally, remark that, for all $i\in \{0,\ldots, 2^m-1\}$, it takes exactly $3m-1$ applications of $B$ to construct $e_{S_{m-1}+i}$, so whenever we multiply by a matrix of $\Sigma$ one more time we need to apply $B$ $3m$ more times to simulate it. It allows to properly express the joint spectral radius of $\Sigma$ from the growth rate of $(B,v)$.

\begin{lemma}\label{controleBitsLemma}
For any $i\in\{1,\ldots,n\}$,
\begin{itemize}
  \item $A_2(B,e_1+e_{k+i})=\{e_2\}$,
  \item $A_3(B,e_1+e_{k+i})=\{e_3, e_4\}$,
  \item and for all $t\in\{3,4\ldots, 3m\}$,
  \begin{itemize}
    \item if $t$ is not divisible by $3$ then $A_{t}(B,e_1+e_{k+i})=\{0\}$,
    \item if $t=3r$ for some integer $r$, then
    $$A_{t}(B,e_1+e_{k+i})=\{0, e_{S_{r-1}},e_{S_{r-1}+1},\ldots, e_{S_{r}-1}\}\,.$$
  \end{itemize}
\end{itemize}
\end{lemma}
\begin{proof}
We fix $i$ and for all $j$,  we write $A_j=A_j(B,e_1+e_{k+i})$.

We proceed by induction on $j$. One easily verifies, $A_4=A_5=\{0\}$ and $A_6=\{0, e_5,e_6,e_7,e_8\}$. Let $t\in\{7,\ldots,3m \}$ such that the lemma holds for any smaller value.
Let $z\in A_t$, then $z=B(x,y)$ for some $j\in\{1,\ldots,t-1\}$, $x\in A_{t-j}$ and $y\in A_j$.
 If $B(x,y)$ is non-zero then, by definition of $B$,
\begin{enumerate}
  \item at least one of $\{y_{S_{m-1}},y_{S_{m-1}+1},\ldots, y_{S_{m}}\}$ is non-zero,
  \item or $y_3$ or $y_4$ is non-zero.
\end{enumerate}

However, by induction hypothesis, $1.$ is impossible and $2.$ only happens if $y\in A_3$ and $j=3$. Thus, $x\in A_{t-3}$.

If $t$ is not a multiple of $3$, then $t-3$ is not a multiple of $3$ and then, by induction hypothesis, $A_{t-3}=\{0\}$ and $A_{t}=\{0\}$.

If $t=3r$ for some integer $r$, then
\begin{align*}
A_{t}&= A_{3r}= \{0\}\cup\{B(x,y): x\in A_{3(r-1)}, y\in \{e_3,e_4\}\}\\
&=\{0\}\cup\left\{B(x,e_3),B(x,e_4): x\in\{0, e_{S_{r-2}},e_{S_{r-2}+1},\ldots, e_{S_{r-1}-1}\}\right\}\\
&=\{0, e_{S_{r-1}},e_{S_{r-1}+1},\ldots, e_{S_{r}-1}\}
\end{align*}
which concludes our proof.
\end{proof}

We let $\tau:\mathbb{R}^n\rightarrow\mathbb{R}^{k+n}$ be the function that maps any vector $x\in\mathbb{R}^{n}$ to a vector whose $k$ first coordinates are $0$ and whose last $n$ coordinates are a copy of $x$, that is, for all $x\in \mathbb{R}^n$ and $i\in \{1,\ldots, k+n\}$, 
$$\tau(x)_i=\begin{cases}
               0, \quad \quad \text{ if } i\le k\\
               x_{i-k}\quad \text{ otherwise}\,.
            \end{cases}$$ We have $\pi(\tau(x))= x$.
\begin{lemma}\label{eqbilivsli}
Let $i\in\{1,\ldots,n\}$.
For any $t>3m$, if $t-1$ is no divisible by $3m$, then
$A_t=\{0\}$. Moreover, for any $r\ge1$,
$$A_{3mr+1}(B,e_1+e_{k+i})=\{\tau(M e'_i): M\in \Sigma^r\}\cup\{0\}\,.$$
\end{lemma}
\begin{proof}
We fix $i$ and we write for all $j$,$A_j=A_j(B,e_1+e_{k+i})$.

We proceed by induction. Let $t$ be an integer such that the lemma gives the correct value of $A_j$ for any $3m<j<t$, and let us show that the lemma gives the correct value for $A_t$.
 
Let $B(x,y)\in A_t\setminus\{0\}$ for some $j$, $x\in A_{t-j}$ and $y\in A_j$.
By definition of $B$, since $B(x,y)$ is non-zero then either:
\begin{enumerate}
  \item at least one of $\{y_{S_{m-1}},y_{S_{m-1}+1},\ldots, y_{S_{m}-1}\}$ is non-zero and $\pi(x)$ is non-zero,
  \item or $y_3$ or $y_4$ is non-zero and at least one of the first $ S_{m-1}$ coordinates of $x$ is non-zero.
\end{enumerate}

By induction hypothesis, for any $z\in \bigcup_{3m<l<t} A_l$, the first $S_{m-1}$ coordinates of $z$ are $0$. Hence, case 2 implies that $y,x\not\in \bigcup_{3m<l} A_l$. We can thus use Lemma \ref{controleBitsLemma} to deduce $j=3$ and $t-j\le 3(m-1)$ which is a contradiction since it implies $3m\ge t-j+j=t>3m$.

In case 1, the induction hypothesis and the fact that $y\in A_{j}$ imply that $j=3m$. Since $\pi(x)$ is non-zero and $x\in A_{t-j}$ the induction hypothesis and Lemma \ref{controleBitsLemma} imply that $t-j\equiv 1 \mod 3m$ leading to $t\equiv 1 \mod 3m$. Thus, if $t> 3m$ and $t\not\equiv 1 \mod 3m$, $A_t=\{0\}$.
It also implies that for any $r\ge1$,
\begin{align*}
 A_{3mr+1}&=\{B(e_l,z): l\in \{S_{m-1},\ldots, S_{m}-1\}, z \in A_{3m(r-1)+1}\}\cup\{0\}\\
 &=\{\tau(M \pi(z)): M\in \Sigma, z \in A_{3m(r-1)+1}\}\cup\{0\}\,.
\end{align*}

If $r=1$ we immediately obtain
$A_{3m+1}= \{\tau(M e'_i): M\in \Sigma\}\cup\{0\}$ as desired. Otherwise, if $r>1$, we use the induction hypothesis to obtain
\begin{align*}
A_{3mr+1}&=\{\tau(M \pi(z)): M\in \Sigma, z \in \{\tau(N e'_i): N\in \Sigma^{r-1}\}\}\cup\{0\}\\
&=\{\tau(M e'_i): M\in \Sigma^r\}\cup\{0\}
\end{align*}
as desired.
\end{proof}

\begin{lemma}\label{rhoaslambda}
We have the following equality:
$$ \rho(\Sigma)= \max_{1\le i\le n}\lambda(B, e_1+e_{k+i}) ^{3m}\,.$$
\end{lemma}
\begin{proof}
We let $||.||_1$ be the $\ell_1$ vector norm, that is,  for all $x\in \mathbb{R}^n$,
$||x||_1= \sum_{j=1}^n|x_j|$. We recall that the corresponding induced norm over matrices is such that for any matrix $M\in\mathbb{R}^{n\times n}$,
$$||M||_1=\max_{1\le j\le n}\sum_{i=1}^{n} |M_{i,j}|=\max_{1\le j\le n} ||M e'_j||_1\,.$$
In particular, for any positive integer $t$,
\begin{equation*}
  \max_{M\in \Sigma^t} ||M||_1^{1/t}
  = \max_{1\le i\le n}\max_{M\in \Sigma^t} ||M e'_i||_1^{1/t}\,.
\end{equation*}
By Lemma \ref{eqbilivsli}, for any $i\in \{1,\ldots,n\}$ and any positive integer $t$
$$\left\{\pi(z):z\in A_{3tm+1}(B, e_1+e_{k+i})\right\} =\left\{Me'_i: M\in \Sigma^t\right\}\cup\{0\}\,.$$
Since $\Sigma$ is non-empty, the vector $0$ does not matter for the maximum and by substituting in the previous equation we obtain
\begin{equation*}
\max_{M\in \Sigma^t} ||M||_1^{1/t}=\max_{1\le i\le n}\max_{z\in A_{3tm+1}(B, e_1+e_{k+i})} ||\pi(z)||_1^{1/t}\,.
\end{equation*}
Since $||.||_1$ is submultiplicative, by definition, the limit as $t$ goes to infinity of the LHS is the joint spectral radius of $\Sigma$, hence
\begin{align*}
\rho (M)&=\lim_{t\rightarrow\infty}\max_{1\le i\le n}\max_{z\in A_{3tm+1}(B, e_1+e_{k+i})} ||\pi(z)||_1^{1/t}\\
&=\limsup_{t\rightarrow\infty}\max_{1\le i\le n}\max_{z\in A_{3tm+1}(B, e_1+e_{k+i})} ||\pi(z)||_1^{1/t}\,.
\end{align*}
Since the $\max$ ranges over a finite set of elements, the $\max$ and $\limsup$ operators commute to obtain
\begin{equation*}
\rho (M)=\max_{1\le i\le n}\limsup_{t\rightarrow\infty}\max_{z\in A_{3tm+1}(B, e_1+e_{k+i})} ||\pi(z)||_1^{1/t}\,.
\end{equation*}
By Lemma \ref{eqbilivsli}, for any $z\in  A_{3tm+1}(B, e_1+e_{k+i})$, the first $k$ coordinates of $z$ are $0$, which implies $||\pi(z)||_1=||z||_1$.
Let $(p_t(i))_{t\ge1}$ be the sequence such that for all $t\ge1$ and $i\in\{1,\ldots,n\}$, $p_t(i)=\max\limits_{z\in A_{3tm+1}(B, e_1+e_{k+i})} ||z||_1^{1/t}$. Then
\begin{equation}
\rho (M)=\max_{1\le i\le n}\limsup_{t\rightarrow\infty} p_t(i)\,.\label{linkeq}
\end{equation}
Let $(p'_t(i))_{t\ge1}$ be the sequence such that for all $t\ge1$ and $i\in\{1,\ldots,n\}$, $p'_t(i)=\max\limits_{z\in A_{t+1}(B, e_1+e_{k+i})}||z||_1^{3m/t}\,.$
The sequence $(p_t(i))_{t\ge1}$ is a subsequence of $(p'_t(i))_{t\ge1}$. By Lemma \ref{eqbilivsli}, all but finitely many of the extra-terms of this second sequence are $0$ and for all $t\ge1$, $p_t(i)\ge0$. Hence,
\begin{align*}
  \limsup_{t\rightarrow\infty}p_t(i)&=  \limsup_{t\rightarrow\infty}p'_t(i)=  \limsup_{t\rightarrow\infty}\max_{z\in A_{t+1}(B, e_1+e_{k+i})} ||z||_1^{3m/t}\\
  &=  \limsup_{t\rightarrow\infty}\max_{z\in A_{t}(B, e_1+e_{k+i})} ||z||_1^{3m/t}=  \left(\limsup_{t\rightarrow\infty}\max_{z\in A_{t}(B, e_1+e_{k+i})} ||z||_1^{1/t}\right)^{3m}\\
  &= \lambda(B, e_1+e_{k+i}) ^{3m}\,.
\end{align*}
Substituting in \eqref{linkeq}, gives $$\rho (M)= \max_{1\le i\le n}\lambda(B, e_1+e_{k+i}) ^{3m}$$ as desired.
\end{proof}

Notice that given a set of $2^m$ matrices of size $n\times n$, the bilinear map $B$ constructed here is of size $ 2^{m+1} +n$ and can be constructed in time polynomial in $2^m$ and $n$. Moreover, the set of coefficients of $B$ is the union of $\{0,1\}$ and of the set of coefficients of $\Sigma$.
Let us sum up what we showed in this section in one theorem.

\begin{theorem}\label{thmJSRvsGR}
There exists an algorithm that given any integers $n$ and $m$ and any set of $2^m$ matrices $\Sigma\subseteq\mathbb{R}^{n\times n}$ outputs in polynomial time a bilinear map $B:\mathbb{R}^{2^{m+1}+n}\times\mathbb{R}^{2^{m+1}+n}\rightarrow\mathbb{R}^{2^{m+1}+n}$ such that
$$\rho (M)= \max_{1\le i\le n}\lambda(B, e_1+e_{2^{m+1}+i}) ^{3m}\,.$$

Moreover, if all the coefficients in $\Sigma$ are nonnegative rationals (resp. non-negative integer, resp. inside $\{0,1\}$) then so are all the coefficients of $B$.
\end{theorem}

We could have used the same idea that we used to ``choose the matrix'' in our reduction to also choose the initial vector amongst the base $(e'_1,\ldots,e'_n)$. That would allow improving Theorem \ref{thmJSRvsGR}, by saying that we can compute a bilinear system $(B,v)$ such that $\rho(\Sigma)= \lambda(B,v)$ (it would be slightly larger, but still polynomial in size). However, we do not see any application of this stronger result that cannot be deduced from Theorem \ref{thmJSRvsGR}.

\section{Undecidability and NP-completness}\label{secmainres}
 We can now use our reduction to deduce some results regarding the difficulty of computing the growth rate of a bilinear system.
 Using Theorem \ref{thmJSRvsGR} with Theorem \ref{thmundecidableJSR} gives the first result which contains a negative answer to Question \ref{maisnquest}.
 
\begin{theorem}\label{thmundecidableGR}
The problem of determining, given a bilinear system $(B,v)$, if $\lambda(B,v)\le 1$ (or $\lambda(B,v)=1$) is Turing-undecidable.
This result remains true even if $B$ and $v$ contains only nonnegative rational entries.
\end{theorem}
\begin{proof}
For the sake of contradiction suppose that $\lambda(B,v)\le1$ is decidable, then we can use it to decide $\rho(\Sigma)\le1$ for any finite set of matrices.
Let $\Sigma\subseteq\mathbb{R}^{n\times n}$ be a set of $2^{m}$ matrices, then by Theorem \ref{thmJSRvsGR} we can find a bilinear map $B$ such that $\rho(\Sigma)=\max_{1\le i\le n}\lambda(B, e_1+e_{2^{m+1}+i}) ^{3m}$.
Then $\rho(\Sigma)\le1$ if and only if for all $i\in\{1,\ldots, n\}$,
$\lambda(B, e_1+e_{2^{m+1}+i}) ^{3m}\le1$. By assumption, the second part is decidable which implies that we can decide $\rho(\Sigma)\le1$. This contradicts Theorem \ref{thmundecidableJSR} and concludes our proof.
\end{proof}
If we replace $1$ in the statement by $\alpha$ where $\alpha$ is a positive real which is part of the input, then the problem is more general and remains undecidable. If $\alpha$ is fixed and is part of the problem definition, then the problem depends on how the coefficients are given (e.g., if $\alpha$ is chosen to be a real number that is not achievable as the growth rate of a system with rational entries\footnote{We do now know any such real number. However, since there are countably many bilinear systems with rationals entries, most real numbers are not the growth rate of a bilinear system.}, then the problem is decidable for systems with rational entries since the answer is always negative). However, $\lambda(B,v)$ is linear in $v$, that is, $\lambda(B,\frac{v}{\alpha})=\frac{\lambda(B,v)}{\alpha}$. In particular, $\lambda(B,v)<\alpha$ if and only if $\lambda(B,\frac{v}{\alpha})<1$. Hence, if dividing any allowed coefficient by $\alpha$ gives a coefficient that is still allowed, the problem remains undecidable for this $\alpha$.

Using Theorem \ref{thmnonpolyJSR}, we can also deduce a complexity result.

\begin{theorem}\label{thmNPcomplete}
Unless $P=NP$, there is no algorithm that, given a bilinear system $(B,v)$ and a
relative accuracy $\varepsilon$, returns an estimate $\widetilde{\lambda}$ of $\lambda(B,v)$ such that $\left|\widetilde{\lambda}-\lambda\right|\le \varepsilon\lambda$ in a number of steps that is polynomial in the size of $\Sigma$ and $\lambda$. It is true even if $B$ and $v$ have binary entries (i.e., $\{0,1\}$).
\end{theorem}
\begin{proof}
Suppose, for the sake of contradiction, that there is such an algorithm.  Let $\Sigma\subseteq\mathbb{R}^{n\times n}$ be a set of $2^{m}$ matrices and let $\varepsilon$ be a positive real.
By Theorem \ref{thmJSRvsGR}, we can compute in polynomial time a bilinear map $B$ such that $\rho(\Sigma)=\max\limits_{1\le i\le n}\lambda(B, e_1+e_{2^{m+1}+i}) ^{3m}$.
Let for all $i\in\{1,\ldots,n\}$, $\lambda_i=\lambda(B, e_1+e_{2^{m+1}+i}) ^{3m}$. Then
   $\rho(\Sigma)=\max\limits_{1\le i\le n}\lambda_i$.

Our assumption implies that we can find in polynomial time for all $i$, $\widetilde{\lambda_i}$ such that  $\left|\widetilde{\lambda_i}-\lambda_i\right|\le \varepsilon\lambda_i$.  Let $\widetilde{\rho}=\max\limits_{1\le i\le n}\widetilde{\lambda_i}$, then
\begin{align*}
\left|\widetilde{\rho}-\rho(\Sigma)\right|
&=\left|\left(\max\limits_{1\le i\le n}\widetilde{\lambda_i}\right)-\left(\max\limits_{1\le i\le n}\lambda_i\right)\right|
=\left|\max\limits_{1\le i\le n}\left(\widetilde{\lambda_i}-\max\limits_{1\le j\le n}\lambda_j\right)\right|\\
  &\le \left|\max\limits_{1\le i\le n}\left(\widetilde{\lambda_i}-\lambda_i\right)\right|\le \max\limits_{1\le i\le n}\left|\widetilde{\lambda_i}-\lambda_i\right|\\
  &\le \max\limits_{1\le i\le n}\varepsilon\lambda_i
  =\varepsilon\max\limits_{1\le i\le n}\lambda_i
  =\varepsilon\rho(\Sigma)\,.
\end{align*}
We obtained a polynomial algorithm to compute an estimate $\widetilde{\rho}$ such that $\left|\widetilde{\rho}-\rho\right|\le \varepsilon\rho$. This is a contradiction with Theorem \ref{thmnonpolyJSR} and this concludes our proof.
\end{proof}

\section{Conclusion}
We showed that there is no algorithm deciding for any positive real $\alpha$ and any bilinear system $(B,v)$ if $\lambda(B,v)=\alpha$. However, as far as we know it might be possible to decide for any bilinear system $(B,v)$ if $\lambda(B,v)=0$. There is a polynomial-time algorithm to decide whether the joint spectral radius is zero (see \cite[Proposition 1.12]{JSRJunger} or \cite{polyzeroJSR}) and the same could hold for the growth rate of a bilinear system.
\begin{question}
  Is there a (polynomial) algorithm to decide whether the growth rate of a bilinear system is zero? What if the coefficients are all nonnegative ? all positive?
\end{question}

Let us discuss the empty cells of table \ref{CompGrowthTable}.
In \cite{RosMSO}, we showed that the growth rate of nonnegative bilinear systems is upper semi-computable (that is, there exists an algorithm that outputs a sequence of upper bounds converging toward the growth rate). In \cite{buigbm2}, Bui showed that if the bilinear map is positive, then the growth rate is computable (that is, there exists an algorithm that outputs a sequence of upper bounds and a sequence of lower bounds). It seems that the approach from \cite{RosMSO}, can be adapted to the general case, so the growth rate is probably upper semi-computable in the general case. However, it is not clear that even if we restrict ourselves to the nonnegative setting, the growth rate is computable.
\begin{question}
  Is the growth rate of a bilinear system computable ? What if the coefficients are all nonnegative?
\end{question}

On the other hand, we showed that we cannot test the equality with $1$ even with nonnegative rational entries. However, it might be the case that there is a decision algorithm if all the coefficients are positive.

\begin{question}
 Is there an algorithm that decides whether the growth rate of any bilinear system with only positive coefficients is $1$ ? What if the coefficients of the bilinear maps are nonnegative and the coefficients of the vector are positive ?
\end{question}

A set of matrices $\Sigma$ is said to be \emph{mortal} if there is a product of matrices from $\Sigma$ that is equal to zero. Following this, we say that a bilinear system is \emph{mortal} if it can produce the vector zero. It is known that the mortality problem is undecidable even for pairs of matrices (see \cite[Corollary 2.2]{JSRJunger}). However, our reduction is not useful in this case, so we leave the question open.
Notice that it is not the same thing as asking if there is a way to combine $n-1$ instances of a bilinear map $B$ into the $n$-linear map that maps everything to zero. This question is probably worth investigating as well.
\begin{question}
  Is there an algorithm to decide whether a bilinear system is mortal, that is, to decide if this bilinear system can produce the vector zero?
\end{question}

Let us conclude with the following remark.
It seems natural to consider the notion of growth rate with a set of bilinear operators and a set of vectors. We could call this the \emph{joint growth rate} of the pair of sets and it might be a useful notion. However, from the point of view of computability having more than one bilinear map or more than one vector does not make things significantly harder. Indeed, we can use a reduction similar to our reduction to reduce this to the computation of the growth rate of a well-chosen bilinear system.

Suppose that we have $m$ bilinear maps $B_1,\ldots,B_m\in\mathbb{R}^n$, then we build one bilinear map $B$ over $\mathbb{R}^{k}$ where $k=2\lceil\log m\rceil + (m+1)n$. We will use the first $2\lceil\log m\rceil$ coordinates of the vector to construct some $e_i$ as we did in our reduction, the next $n$ coordinates are where we simulate our computation (``the simulated value''), and the remaining $mn$ coordinates are used for temporary computation ($m$ blocks of size $n$). We moreover want at most one of these three parts to be non-zero for any vector. For any vectors $x$, $y$ and for all $i$ the $i$th temporary block of $n$ coordinates of $B(x,y)$ is the result of the $i$th bilinear map applied to the ``simulated values'' of $x$ and $y$. Moreover, we want that for all $i$, the ``simulation value'' of $B(x,e_{S_{n-1}+i})$ contains the $i$th temporary block of $x$. Now if we can apply $B$ to any pairs of vectors with meaningful simulated values and then apply $B$ to the result and to $e_{S_{n-1}+i}$ and it simulates the application of $B_i$ to the two simulated values of $x$ and $y$.
For instance, for the bilinear system $(\{B_1,B_2\},v)$ we could use
\[B(x,y)=\begin{pmatrix}
           0\\x_1y_1\\x_1y_2\\x_2y_1\\
           x_3y_{[5+n,\ldots,4+2n]}+x_4y_{[5+2n,\ldots,4+3n]}\\
          B_1(x_{[5,\ldots,4+n]},y_{[5,\ldots,4+n]})\\
          B_2(x_{[5,\ldots,4+n]},y_{[5,\ldots,4+n]})\\
         \end{pmatrix}\, \text{ and } 
         v'=\begin{pmatrix}
              1\\0\\0\\0\\v\\0\\\vdots\\0
            \end{pmatrix}
\,.
\]

 \end{document}